\documentclass[10pt]{IEEEtran}

\usepackage{amsfonts}
\usepackage{amssymb}
\usepackage{stfloats}
\usepackage{cite}
\usepackage{graphicx}
\usepackage{psfrag}
\usepackage{subfigure}
\usepackage{amsmath}
\usepackage{array}
\usepackage{epstopdf}
\usepackage{authblk}
\usepackage{graphicx} 
\usepackage{amsthm} 
\usepackage{lipsum}
\usepackage{verbatim} 
\usepackage{authblk}
\usepackage{algorithm}
\usepackage{algpseudocode}
\usepackage{amsmath}
\usepackage{graphics}
\usepackage{epsfig}
\usepackage{setspace}

\usepackage[usenames]{color}

\newtheorem{Definition}{Definition}
\newtheorem{Problem}{Problem}
\newtheorem{Lemma}{Lemma}

\newtheorem{Solution}{Solution}
\newtheorem{Baseline}{Baseline}

\title{Cellular Offloading via Downlink Cache Placement}

\author{Bojie Lv, Lexiang Huang and Rui Wang\\
The Southern University of Scienece and Technology\\
Email: \{lvbj, huanglx\}@mail.sustc.edu.cn, wang.r@sustc.edu.cn} 
 
\begin{document}

\maketitle
\begin{abstract}
In this paper, the downlink file transmission within a finite lifetime is optimized with the assistance of wireless cache nodes. Specifically, the number of requests within the lifetime of one file is modeled as a Poisson point process. The base station multicasts files to downlink users and the selected the cache nodes, so that the cache nodes can help to forward the files in the next file request. Thus we formulate the downlink transmission as a Markov decision process with random number of stages, where transmission power and time on each transmission are the control policy. Due to random number of file transmissions, we first proposed a revised Bellman's equation, where the optimal control policy can be derived. In order to address the prohibitively huge state space, we also introduce a low-complexity sub-optimal solution based on an linear approximation of the value function. The approximated value function can be calculated analytically, so that conventional numerical value iteration can be eliminated. Moreover, the gap between the approximated value function and the real value function is bounded analytically. It is shown by simulation that, with the approximated MDP approach, the proposed algorithm can significantly reduce the resource consumption at the base station.
\end{abstract}

\section{introduction}\label{sec:intro}
Caching is a promising technology for future cellular networks, which could improve the network spectral efficiency \cite{Leung2014} or cut down energy consumption \cite{ISWCS,Vincent2013} by exploiting the repeated transmissions of the same content. With the wired connection between cache nodes and base station (BS), the buffer capacity limitation at the cache nodes becomes the major performance bottleneck, and there have been a number of research efforts spent on the file placement of the cache nodes. For example, in order to offload more work to cache nodes, a file placement algorithm was derived in \cite{Leung2016}. A long-term file placement policy was proposed in \cite{Vincent2013} to minimize the transmit power in BS. In \cite{Debbah2014}, it is shown that prediction of future demand will improve the performance of file placement at the cache nodes.

All the above works assume that there are wired links between the cache nodes and BSs. It might be costly to deploy cache nodes with wired connection in some areas, and hence the cache placement via wireless links (e.g., cellular downlink) should also be studied \cite{Survey2015}. In \cite{WCNC2017}, the transmissions from BS to cache nodes and from cache nodes to users share the same spectrum. However, it neglects the possibility that both cache nodes and users can listen to the BS simultaneously via a multicast mode. In fact, if file caching has to be made via downlink, the phase of cache placement can be coupled with the phase of serving requesting users. For example, at the first few transmissions of one file, both cache nodes and requesting users can listen to the BS simultaneously; and the cache nodes can help to forward the file as long as they have been able to decode it. Thus the downlink strategy should be optimized spanning the whole lifetime of a file (includes the transmission to both cache nodes and users). However, this has not been addressed by the existing literature.

In this paper, we would like to shed some light on the above open issue by considering the downlink file transmission with the assistance of cache nodes, which can only receive data via downlink transmission. The communication links between cache nodes and users are via different spectrum from the downlink (e.g., Wi-Fi) as \cite{May2016,Molisch2016}. Hence the BS try to minimize the average downlink resource consumption by offloading some of the traffic to cache nodes. Specifically, we model the downlink transmission of one file within his lifetime as a Markov decision process (MDP) with random number of stages. Note that this is not conventional MDP problem with finite and fixed number of stages, we propose a revised version of Bellman's equation, where the optimal control policy can be obtained given the value function. Then we introduce a linear approximation on the value function so that the exponential complexity can be reduced into linear. The bounds on the approximation error are also obtained. It is shown by simulations that the proposed scheme can significantly suppress the overall transmission consumption compared with baselines. 

The remainder of this paper is organized as follows. In Section \ref{sec:model}, the system model is introduced. In Section \ref{sec:formulation}, we formulate the downlink resource allocation as a MDP with random number of stages. In Section \ref{sec:optimal}, the approach to obtain the optimal control policy is explained. In Section \ref{sec:approximation}, a linear approximation is proposed to the value function, and the bound on approximation error is derived. The numerical simulation is provided in Section \ref{sec:sim} and the conclusion is drawn in Section \ref{sec:con}.

\section{System Model}\label{sec:model}

\subsection{Downlink File Request Model}

The downlink transmission in a cell with one multi-antenna BS, $N_C $ single-antenna cache nodes is considered. Let $ N_T $ be the number of antennas at the BS. Without loss of generality, it is assumed that the BS locates in the origin, and the locations of cache nodes are denoted as $ \mathbf{c}_1, \mathbf{c}_2, ..., $ and $ \mathbf{c}_{N_C} $ respectively, which can be arbitrary in the cell coverage. Let $ \mathcal{C}_i $ ($ \forall i=1,2,...,N_C $) be the coverage region of the $ i $-th cache node.

The downlink data is organized by files. Similar to most of the existing literature, it is assumed that each file consists of $ R_f $ information bits, and it is further divided into $ N_S $ segment equally. These files can be popular videos, breaking news or other web content, which might be requested by multiple users in the cell. Note that the popularity of web content depends heavily on the time. For example, breaking news may not be popular anymore after one day. To model the dynamics in file popularity, it is assumed that each file has a lifetime $ T $. One file will be downloaded in the cache nodes only during its lifetime. For the elaboration convenience, we consider the same file size and lifetime. In fact, our approach is applicable for heterogeneous file sizes and lifetimes.

Users appear randomly in the cell coverage to request files via BS. Suppose the $ f $-th file ($ f=1,2,... $) is available for access since time instance $ t_f $, we consider the requests on this file happen during the time period $ [t_f, t_f+T] $. The Poisson point process is adopted to model the event of file request within a file's life time, and $ \lambda_f $ is the process intensity of the each file. Note that in queueing theory, Poisson point process is widely accepted to model random arrival events, such as the arrival of customers at a store or phone calls at an exchange. Hence, the probability mass function of request number for the $ i $-th file ($ \forall i $), when the remaining life time is $ T_{rem} $, is given by
\begin{equation} \label{eqn:request}
\Pr(\mbox{Request number} = n) = \frac{(\lambda_f T_{rem})^n}{n!} e^{-\lambda_f T_{rem}}.
\end{equation}

\subsection{Wireless Caching Model}

We consider the scenario that wired connections between BS and cache nodes are not available, hence the cache nodes can only receive the files from the downlink. Compared with wired connection or dedicated spectrum for the communications between cache nodes and BS, this scenario provides more flexibility on deployment. As a result, there are two modes of data transmission in the network.

$ \bullet $ {\bf Downlink multicast:} the BS delivers the requested file segments to user and some of the cache nodes. This will happen when the requested segments cannot be found in the cache nodes nearby. For example, if one user is within the coverage of $ \mathcal{C}_i $, it will receive the file segments from the BS when these file segments cannot be found in the $ i $-th cache node. Since both requesting user and cache nodes can listen when the BS is transmitting, the transmission mode is multicast.
	
$ \bullet $ {\bf Device-to-device (D2D) communications:} the cache nodes forward the requested file segments to users. This will happen when the requested segments hit the buffer of nearby cache nodes. For example, if one user is within the coverage of $ \mathcal{C}_i $, it will receive the file segments from the $ i $-th cache node when the file segments can be found there. This communication is made directly from cache nodes to users. It can use Wi-Fi, bluetooth, or other air interfaces, which is not in the same spectrum as downlink. For example, the offloading from cellular to Wi-Fi has attracted a number of attenuation recently \cite{May2016} \cite{Molisch2016}.

Since the second mode of communications can be done distributively and parentally with relatively low transmission power, we consider the transmission resource consumption at the BS only, which is the bottleneck of the overall system. 

The first transmission of each file is made directly from the BS, while some cache nodes may also be able to decode the whole file or some segments. Hence, for the following requests of the same file, user will receive one segment from the cache node if the following two conditions are satisfied: (1) the user is in the coverage of certain cache node; (2) the aforementioned file segment has been successfully decoded by the aforementioned cache node. For elaboration convenience, we shall refer to the user, which sends the $ n $-th request on the $ f $-th file, as the $ (f,n) $-th user, and refer to the $ s $-th segment of the $ f $-th file as the $ (f,s) $-th segment. Since the data rate of wireless communications is usually much smaller than wired, and each file has finite lifetime, we ignore the limitation on cache buffer in this paper \footnote{For example, suppose that one BS is transmitting downlink files with overall data rate of $ 1 $ Gbps, and the lifetime of each file is $ 24 $ hours. Then the maximum required storage capacity of one cache node is around $ 10 $T bytes, which is a mild requirement for ignoring the buffer capacity limitation. }.

\subsection{Downlink Physical Layer Model}

In downlink, the receivers include the requesting user and cache nodes, and the space-time block code (STBC) with full diversity is used at the BS to facilitate the multicast communications. The benefits of STBC in  multicast are that (1) the BS need not to collect channel state information (CSI); (2) the full diversity can be achieved at all the receivers.

Since the transmission time of one file segment is much larger than the channel coherent time, it is assumed that the ergodic channel capacity span all possible small-scale channel fading can be achieved during one segment transmission. Let $ \rho_{f,n} $ and $ \rho_c $ be the pathloss from the BS to the $ (f,n) $-th user and the $ c $-th cache node respectively, $ \eta_{f,n,s} $ and $ \eta_{f,n,s}^c $ be the corresponding shadowing attenuation in $ n $-th transmission of the $ (f,s) $-th segment, $ P_{f,n,s} $ be the downlink transmission power of the $ s $-th file segment in response to request of the $ (f,n) $-th user, $ N_{f,n,s} $ be the number of downlink transmission symbols scheduled to deliver the $ s $-th segment to the $ (f,n) $-th user, the throughput achieved by the $ (f,n) $-th downlink user in the transmission of the $ s $-th segment is given by
\begin{equation}\label{eqn:dl-rate}
R_{f,n,s} = N_{f,n,s} \mathbb{E}_{\mathbf{h}_{f,n,s}} \left[ \log_2 \left( 1 + \frac{||\mathbf{h}_{f,n,s}||^2 P_{f,n,s}}{N_T \sigma^2_z} \right) \right],
\end{equation}
where $ \sigma_z^2 $ is the power of noise as well as inter-cell interference, $ \mathbf{h}_{f,n,s} $ is the i.i.d. (independently and identically distributed) channel vector from the BS to the request user. Each element of $ \mathbf{h}_{f,n,s} $ is complex Gaussian distributed with zero mean and variance $ \rho_{f,n}\eta_{f,n,s} $. As a remark note that the transmission of one segment may consume a large number of frames, and the channel vector $ \mathbf{h}_{f,n,s} $ can be different from frame to frame. However, since we consider the ergodic channel capacity, the randomness in small-scale fading is averaged. Hence, the $ (f,n) $-th user can decode the $ s $-th segment only when
\begin{equation}\label{eqn:cont}
R_{f,n,s}  \geq R_f/N_S.
\end{equation}

Moreover, the achievable data rate of the $ c $-th cache node is given by
\begin{equation}\label{eqn:dl-cache}
R_{f,n,s}^c = N_{f,n,s} \mathbb{E}_{\mathbf{h}_{f,n,s}} \left[ \log_2 \left( 1 + \frac{||\mathbf{h}_{f,n,s}^c||^2 P_{f,n,s}}{N_T \sigma^2_z} \right) \right],
\end{equation}
where $ \mathbf{h}_{f,n,s}^c $ is the i.i.d. channel vector from the BS to $ c $-th cache node. Each element of $ \mathbf{h}_{f,n,s}^c $ is complex Gaussian distributed with zero mean and variance $ \rho_{c}\eta_{f,n,s}^c $. The $ c $-th cache node can decode the $ s $-th segment only when $
R_{f,n,s}^c  \geq R_f/N_S. $
In the downlink transmission of one file, the location of requesting user, thus $ \rho_{f,n} $,  is assumed to be static. Moreover, the shadowing effect $ \eta_{f,n,s} $ and $ \eta_{f,n,s}^c $ is assumed to be i.i.d. for different segments.

\subsection{Control Policy}

There is an scheduling issue in the downlink transmission: if more power or time resource is spent in downlink file transmission, more cache nodes are able to buffer the file (or some segments of the file), which may save the resource of the BS in the successive transmission of the same file. Thus it is necessary to optimize the resource allocation during the whole lifetime of one file, instead of single transmission. Since this is a multi-stage optimization problem, we shall formulate it as finite-stage MDP. In this section, the system state and control policy will be defined. 

\begin{Definition}[System State]
	Before the transmission to the $ (f,n) $-th user, the system status is uniquely specified by $ S_{f,n} = \left[\mathcal{B}_{f,s}^c,\rho_{f,n} , \eta_{f,n,s}, \eta_{f,n,s}^c| \forall c=1,...,N_C; s=1,...,N_S  \right]$, where $ \mathcal{B}_{f,s}^c =1 $ means that the $ s $-th segment of the $ f $-th file has been successfully decoded by the $ c $-th cache node and $ \mathcal{B}_{f,s}^c =0 $ means otherwise.
\end{Definition}

Then the downlink scheduling policy is defined as follows.
\begin{Definition}[Control Policy]
	Suppose the $ s $-th segment is transmitted to the $ (f,n) $-th user via downlink. Given the system state $ S_{f,n} $, the scheduling policy $ \Omega_{f,n} $ ($ \forall f,n $) is a mapping from system state and the remaining lifetime $ T_{f,n} $ to the scheduling parameters $ P_{f,n,s} $ and $ N_{f,n,s} $ ($ \forall s $). Thus $\Omega_{f,n}( S_{f,n}, T_{f,n})= \{(P_{f,n,s}, N_{f,n,s}) | \forall s\} $. Moreover, in order to guarantee the requesting user can successfully decode the downlink data, the constraint in (\ref{eqn:cont}) should be satisfied.
\end{Definition}

\section{Problem Formulation} \label{sec:formulation}
In this section, we shall formulate the downlink resource control problem as a MDP with random number of stages. Let $ \mathcal{C}_{f,n}^s = \cup_{\forall i, \mathcal{B}_{f,s}^i=1} \mathcal{C}_i $ be the area where the requesting users is able to receive the $ s $-th segment of the $ f $-th file from one of the cache nodes, and $ \mathbf{l}_{f,n} $ be the location of the $ (f,n) $-th user. Hence we use the following cost function to measure the weighted sum of energy and transmission time of the BS, which is spent on the $ (f,n) $-th user.
\begin{equation}\nonumber
g_{f,n,s} = I(\mathbf{l}_{f,n} \notin \mathcal{C}_{f,n}^s) \times (w_e P_{f,n,s}N_{f,n,s} + w_t N_{f,n,s}),
\end{equation}
where $ w_e $ and $ w_t $ are the weights on transmission energy and transmission opportunities respectively, and $ I(\cdot) $ is the indication function. Since transmission from cache nodes to users is made via other air interfaces, the transmission resource used by the cache nodes is not counted in this cost function. Hence the average cost spent on the $ f $-th file is given by
\begin{equation}\nonumber
\overline{g}_{f} \left( \{\Omega_{f,n}|\forall n\} \right) = \sum_N \mathbb{E}_{\eta,\rho} \left[ \frac{(\lambda_f T)^N}{N!} e^{-\lambda_f T} \sum_{n=1}^{N} \sum_{s=1}^{N_S} g_{f,n,s}\right],
\end{equation}
where the expectation is taken over all possible large-scale channel fading (including the shadowing effect $ \eta $ and pathloss $ \rho $) in the system. The summation on $ N $ is due to the random number of requests as elaborated in (\ref{eqn:request}).

As a result, the overall system cost function is given by
\begin{equation} \nonumber
\overline{G}(\{\Omega_{f,n}|\forall f,n\}) = \lim_{F \rightarrow +\infty} \frac{1}{F}\sum_{f=1}^{F} \overline{g}_{f} \left( \{\Omega_{f,n}|\forall n\} \right),
\end{equation}
and the system optimization problem can be written as

\begin{Problem}[Overall System Optimization] \label{prob:overall}
\begin{eqnarray}
&\min\limits_{\{\Omega_{f,n}|\forall f,n\}} &\overline{G}(\{\Omega_{f,n}|\forall f,n\}) \nonumber\\
&s.t.& (\ref{eqn:cont}),\  \forall f,n,s. \nonumber
\end{eqnarray}
\end{Problem}

Since there is no constraint on the cache nodes' buffer capacity and transmission resources, the above optimization problem can be further decoupled into the following sub-problems with respect to each file.
\begin{Problem}[Optimization on the $ f $-th File]\label{prob:main}
\begin{eqnarray}
&\min\limits_{\{\Omega_{f,n}|\forall n\}} &\overline{g}_f(\{\Omega_{f,n}|\forall n\}) \nonumber\\
&s.t.& (\ref{eqn:cont}), \ \forall n,s. \nonumber
\end{eqnarray}
\end{Problem}

\section{Optimal Control Policy} \label{sec:optimal}

Note that Problem \ref{prob:main} is a dynamic programming problem with random number of stages, which cannot be solved by the standard approach as \cite{MDP1}. We shall shown  in this section that the optimal solution of Problem \ref{prob:main} (i.e. policy iteration) can be obtained by applying value iteration on another MDP problem (Problem \ref{prob:fix} as follows)  with finite and fixed number of stages first, then solving a revised version of Bellman's equation. First of all, we introduce the following MDP problem with fixed number of stages.

\begin{Problem}[Optimization with Fixed Stage Number] \label{prob:fix}
	\begin{eqnarray}
	&\min\limits_{\{\Omega_{f,n}|\forall n\}} &\mathbb{E}_{\eta,\rho} [\sum_{n=1}^{N_R}\sum_{s=1}^{N_S} g_{f,n,s}]\nonumber\\
	&s.t.& \mbox{(\ref{eqn:cont})}, \ \forall n,s. \nonumber
	\end{eqnarray}
	where $ N_R $ is the fixed number of requests on the $ f $-th file.
\end{Problem}

As introduced in \cite{MDP1}, there is standard solution for the MDP problem with finite and fixed number of stages. The optimal solution of Problem \ref{prob:fix} can be deduced via the Bellman's equation in (\ref{eqn:bellman-fix}) on the top of next page, where $ V_{N_R-n+1}(S_{f,n}) $ is usually named as value function of the $ n $-th stage, and $ S_{f,n+1} $ denotes the next state of the $ f $-th file. In fact, $ V_{N_R-n+1}(S_{f,n}) $ means the average remaining cost of the $ f $-th file from the $ n $-th transmission to the $ N_R $-th transmission, given the system state of the $ n $-th stage $ S_{f,n} $.

Note that the large-scale fading is i.i.d. in each file transmission, the expectation on large-scale fading can be taken on both side of the above equation. Hence we have the following conclusion, whose proof is straightforward and neglected here.

\begin{Lemma}[Bellman's Equation with Reduced Space]
The optimal control policy of Problem \ref{prob:fix} is the solution of the Bellman's equation with reduced state space in (\ref{eqn:bellman-reduce}), where $\widetilde{S}_{f,n}= \{\mathcal{B}_{f,s}^c \in S_{f,n}|\forall c,s \}$, $\widetilde{V}_{N_R-n}(\widetilde{S}_{f,n})=\mathbb{E}_{\eta,\rho}[V_{N_R-n}(S_{f,n+1})] $, and $ \Omega_{f,n}(\widetilde{S}_{f,n}) = \{\Omega_{f,n}(S_{f,n})|\forall \rho_{f,n} , \eta_{f,n,s}, \eta_{f,n,s}^c, c \} $.
\end{Lemma}

The standard value iteration can be used to solve the Bellman's equation (\ref{eqn:bellman-reduce}), and obtain the value function $ \widetilde{V}_{N_R-n+1}(\widetilde{S}_{f,n}) $ ($ \forall n $). In the following lemma, we show that the optimization problem of this paper (Problem \ref{prob:main}) can also be solved via the above value function. 

\begin{Lemma}[Optimal Control Policy of Problem 1] With the value function $ \widetilde{V}_{N_R-n+1}(\widetilde{S}_{f,n}) $ ($ \forall n $), the optimal control policy for Problem \ref{prob:main}, denoted as $ \Omega_{f,n}^*(\widetilde{S}_{f,n},T_{f,n}) $, can be calculated from (\ref{eqn:bellman-random}).
\end{Lemma}
\begin{proof} 
Equation (\ref{eqn:bellman-random}) is straightforward due to the factor that (1) $ \widetilde{V}_{N}(\widetilde{S}_{f,n+1}) $ denotes the averaged cost for $ N $ times of requests on the $ f $-th file, give the system state $ \widetilde{S}_{f,n+1} $; (2) $ \frac{(\lambda_f T_{f,n})^N}{N!} e^{-\lambda_f T_{f,n}} $ is the probability that there are $ N $ times of file requests within the duration $ T_{f,n} $.
\end{proof}

As a result, the optimal control policy of the Problem \ref{prob:main} can be solved via the following two steps.
\begin{itemize}
	\item \textbf{Value Iteration}: Calculate the value function $ \widetilde{V}_{N_R-n+1}(\widetilde{S}_{f,n})( \forall n, \widetilde{S}_{f,n} ) $  via Bellman's equation with reduced space (\ref{eqn:bellman-reduce}).
	\item \textbf{Policy Iteration}: Obtain optimal control policy from (\ref{eqn:bellman-random}).
\end{itemize}

Note that the state space of the system is actually huge. For example, suppose there are 20 cache nodes in the system and 50 segments per file, the dimensional of system state is $2^{1000}$, which is referred to as the curse of dimensionality. The optimal solution is actually computationally intractable. Hence, we continue to propose a low-complexity solution based on the technique of approximated MDP in the following section.

\begin{figure*}

\begin{eqnarray}
V_{N_R-n+1}(S_{f,n})=\min_{\Omega_{f,n}(S_{f,n})} \bigg\{\sum_{s} g_{f,n,s}(S_{f,n},\Omega_{f,n}) +\sum\limits_{S_{f,n+1}}{V_{N_R-n}(S_{f,n+1})Pr(S_{f,n+1}|S_{f,n},\Omega_{f,n})} \bigg\} \label{eqn:bellman-fix}
\end{eqnarray}
\begin{eqnarray}\label{eqn:bellman-reduce}
\widetilde{V}_{N_R-n+1}(\widetilde{S}_{f,n})=\min_{\Omega_{f,n}(\widetilde{S}_{f,n})}  \mathbb{E}_{\eta,\rho} \bigg\{ \sum_{s} g_{f,n,s}({S}_{f,n},\Omega_{f,n}) +\sum\limits_{\widetilde{S}_{f,n+1}}{\widetilde{V}_{N_R-n}(\widetilde{S}_{f,n+1})Pr(\widetilde{S}_{f,n+1}|{S}_{f,n},\Omega_{f,n})} \bigg\}
\end{eqnarray}
\begin{eqnarray}\label{eqn:bellman-random}
\Omega_{f,n}^*(\widetilde{S}_{f,n},\!T_{f,n})\!=\!\arg\min_{\Omega_{f,n}(\widetilde{S}_{f,n})} \!\!\!\!\! \mathbb{E}_{\eta,\rho} \bigg\{\sum_{s} g_{f,n,s}({S}_{f,n},\!\Omega_{f,n})  +\!\!\!\! \!\sum\limits_{\widetilde{S}_{f,n+1},N}\!\!\!\!\!\!\!\frac{(\lambda_f T_{f,n})^N}{N!} e^{-\lambda_f T_{f,n}}{\widetilde{V}_{N}(\widetilde{S}_{f,n+1})Pr(\widetilde{S}_{f,n+1}|{S}_{f,n},\!\Omega_{f,n})}\! \bigg\}
\end{eqnarray}
\begin{equation}\label{eqn:app-v}
\widetilde{V}_{N_R-n+1}(\widetilde{S}_{f,n}) \approx \widetilde{V}_{N_R-n+1}(\widetilde{S}_{f}^*) + \sum_{\{(i,s)|\forall \mathcal{B}^i_{f,s}(\widetilde{S}_{f,n}) =0\}} \bigg( \widetilde{V}_{N_R-n+1}(\widetilde{S}_{f}^{i,s})-\widetilde{V}_{N_R-n+1}(\widetilde{S}_{f}^*)\bigg)
\end{equation}

\end{figure*}

\section{Low-Complexity Solution via Approximated MDP}\label{sec:approximation}

In this section, we shall first introduce a novel linear approximation on the value function $ \widetilde{V}_{N_R-n+1}(\widetilde{S}_{f,n}) $, and elaborate the on-line control policy to determine the control actions given the current system state and approximated value function.

\subsection{Approximation on Value Function}

We first define the notations for following reference system states.
\begin{itemize}
	\item $ \widetilde{S}_{f}^* $ is the state of $f$-th file where all the cache nodes have successfully decoded the whole file. Thus $
	\widetilde{S}_{f,n}^*= \{\mathcal{B}_{f,s}^{c} =1| \forall c=1,...,N_C\} .$
	
	\item $ \widetilde{S}_{f}^{i,s} $ is the state of $ f $-th file transmission where only the $ s $-th segment at the $ i$-th cache node is not successfully decoded. $
	\widetilde{S}_{f}^{i,s}= \{\mathcal{B}_{f,s}^{i} = 0, \mathcal{B}_{f,t}^{j}=1| \forall (j,t) \neq (i,s) \} .$
\end{itemize}
Hence, we approximate the value function $ \widetilde{V}_{N_R-n+1}(\widetilde{S}_{f,n}) $  as (\ref{eqn:app-v}), where $\mathcal{B}^i_{f,s}(\widetilde{S}_{f,n})  $ means the parameter of $ \mathcal{B}^i_{f,s} $ in the system state $ \widetilde{S}_{f,n} $.

In order to apply this approximation on all value function, it is necessary to obtain the value of $\widetilde{V}_{N_R-n+1}(\widetilde{S}_{f}^*)$ and $\widetilde{V}_{N_R-n+1}(\widetilde{S}_{f}^{i,s})$ for all $ n,i $, and $ s $ via (\ref{eqn:bellman-reduce}). Generally, they can be evaluated via numerically simulation. In the following, however, we provide the analytically expressions for them.

\subsubsection{Evaluation of $ \widetilde{V}_{N_R-n+1}(\widetilde{S}_{f}^*) $}

Note that system state $ \widetilde{S}_{f}^* $ represents the situation that all the cache nodes have already decoded the $ f $-th file, the purpose of downlink transmission is only to make sure that the requesting users, which are outside of the coverage of cache nodes, can decode the downlink file. Hence it is clear that
\begin{eqnarray}\nonumber
&&\widetilde{V}_{N_R-n+1}(\widetilde{S}_{f}^*) \nonumber\\
&=& (N_R-n+1)\Pr(\mathbf{l}_{f,n} \notin \mathcal{C}_{f,n}^s) \nonumber\\
&&\mathbb{E}_{\rho,\eta} \bigg[\min \limits_{[P_{f,n,s} \atop N_{f,n,s}]}  w_e P_{f,n,s}N_{f,n,s} + w_t N_{f,n,s} |\mathbf{l}_{f,n} \notin \mathcal{C}_{f,n}^s\bigg] \nonumber \\
&s.t.& R_{f,n,s} \geq R_f/N_S, \ \forall s.\nonumber
\end{eqnarray}

The above value function can be calculated with analytical expression, which is elaborated below.
\begin{Lemma}\label{Lemma:Segment}
The value function $ \widetilde{V}_{N_R-n+1}(\widetilde{S}_{f}^*) $ is given by
\begin{eqnarray}\nonumber
\widetilde{V}_{N_R-n+1}(\widetilde{S}_{f}^*) 
\approx (N_R-n+1)\Pr(\mathbf{l}_{f,n} \notin \mathcal{C}_{f,n}^s) \nonumber\\
\mathbb{E}_{\rho,\eta} \bigg[\sum_s w_e P_{f,n,s}^*N_{r,f,s}^* + w_t N_{f,n,s}^* |\mathbf{l}_{f,n} \notin \mathcal{C}_{f,n}^s \bigg], \nonumber 
\end{eqnarray}
where $
P_{f,n,s}^* =\frac{w_t} {w_e \mathbb{W}(\frac{2^{\theta}w_t}{ew_e}) },
N_{f,n,s}^* =\frac {R_f}{N_S[\theta+\log_2(P_{f,n,s}^*)]} $, $
\theta =	\mathbb{E}_{\mathbf{h}_{f,n,s}} \left[ \log_2 \left(  \frac{||\mathbf{h}_{f,n,s}||^2}{N_T \sigma^2_z}\right) \right] $, 
and $\mathbb{W}(x)$ is the Lambert-W function \cite{W}.
\end{Lemma}

\begin{proof}
	Please refer to Appendix A.	
\end{proof}
\subsubsection{Evaluation of $\widetilde{V}_{N_R-n+1}(\widetilde{S}_{f}^{i,s})$}

Given system state $ \widetilde{S}_{f}^{i,s} $ for arbitrary stage, there are only two possible next system states $  \widetilde{S}_{f}^{i,s} $ and $  \widetilde{S}_{f}^{*} $, which are discussed below.

$ \bullet $ When $ \rho_{f,n}  \eta_{f,n,s} \leq \rho_i \eta_{f,n,s}^i $, thus $ R_{f,n,s} \leq R_{f,n,s}^i$, the $ i $-th cache node is alway able to decode the $ s $-th file segment give that the transmission constraint (\ref{eqn:cont}) should be satisfied. Thus the next state must be $  \widetilde{S}_{f}^{*} $. In this case, the optimized RHS of (\ref{eqn:bellman-reduce}) is given by
	\begin{eqnarray} 
	Q_c = \mathbb{E}  \min\limits_{\Omega_{f,n}(S_{f,n})} \bigg\{ \sum_{s} g_{f,n,s}({S}_{f,n},\Omega_{f,n})\bigg\} +{\widetilde{V}_{N_R-n}(\widetilde{S}_{f}^{*})},  \nonumber
	\end{eqnarray}
	subject to constraint (\ref{eqn:cont}).
	
$ \bullet $ When $ \rho_{f,n}  \eta_{f,n,s} > \rho_i \eta_{f,n,s}^i $, thus $ R_{f,n,s} > R_{f,n,s}^i$, the BS can choose to secure the transmission of the $ s $-th segment to the $ (f,n) $-th user or the $ i $-th cache node. Hence the optimized  RHS of (\ref{eqn:bellman-reduce}) is given by $
	\mathbb{E}  \min \bigg \{Q_u({S}_{f,n}),Q_{i,s}({S}_{f,n}) \bigg\}$, 
	where $Q_u$ and $Q_{i,s}$ are defined in (\ref{eqn:a_u}) and (\ref{eqn:a_i}) respectively.

As a result, the expression of $\widetilde{V}_{N_R-n+1}(\widetilde{S}_{f}^{i,s})$ is summarized by the following lemma.
\begin{Lemma} The value function  $\widetilde{V}_{N_R-n+1}(\widetilde{S}_{f}^{i,s})$ is given by (\ref{eqn:v_default one fragment }). Moreover, the optimal control actions for $ Q_c $ and $ Q_u $ are the same as Lemma \ref{Lemma:Segment}. The optimal control action for  $ Q_{i,s} $ is given by $
	P_{f,n,s}=\frac{w_t} {w_e \mathbb{W}(\frac{2^{\theta^i}w_t}{ew_e}) },N_{f,n,s}=\frac {R_f}{N_S[\theta^i+\log_2(P_{f,n,s})]} $, $ \theta^i=	\mathbb{E}_{\mathbf{h}_{f,n,s}^i} \left[ \log_2 \left(  \frac{||\mathbf{h}_{f,n,s}^i||^2}{N_T \sigma^2_z}\right) \right]$; and $ \forall t\neq s $, $
P_{f,n,t}^* =\frac{w_t} {w_e \mathbb{W}(\frac{2^{\theta}w_t}{ew_e}) },
N_{f,n,t}^* =\frac {R_f}{N_S[\theta+\log_2(P_{f,n,t}^*)]}$, $
\theta =	\mathbb{E}_{\mathbf{h}_{f,n,t}} \left[ \log_2 \left(  \frac{||\mathbf{h}_{f,n,t}||^2}{N_T \sigma^2_z}\right) \right] $.
\end{Lemma}

\begin{proof}
The proof is similar to that of Lemma \ref{Lemma:Segment}, and it is omitted here.	
\end{proof}

\begin{figure*}

\begin{eqnarray}\label{eqn:a_u}
	&Q_u({S}_{f,n})= & \min\limits_{\Omega_f(S_{f,n})} \sum_{t} g_{f,n,t}({S}_{f,n},\Omega_{f,n}) +{\widetilde{V}_{N_R-n}(\widetilde{S}_{f}^{i,s})},\ \ s.t. \ \ R_{f,n,t}=R_f/N_S, \  \forall t
\end{eqnarray}
\begin{eqnarray}\label{eqn:a_i}
&Q_{i,s}({S}_{f,n})= & \min\limits_{\Omega_f(S_{f,n})} \sum_{t} g_{f,n,t}({S}_{f,n},\Omega_f) +{\widetilde{V}_{N_R-n}(\widetilde{S}_{f}^{*})},\ \ s.t. \ \  R_{f,n,s}^i = R_f/N_S \mbox{ and } R_{f,n,t}=R_f/N_S, \  \forall t
\end{eqnarray}
\begin{eqnarray}\label{eqn:v_default one fragment }
\widetilde{V}_{N_R-n+1}(\widetilde{S}_{f}^{i,s}) \!=\! \mathbb{E}_{\eta,\rho}[Q_c|R_{f,n,s} \!\leq\! R_{f,n,s}^i]Pr(R_{f,n,s} \!\leq\! R_{f,n,s}^i) \!+\!\mathbb{E}_{\eta,\rho}[\min \{Q_u, Q_{i,s}\}|R_{f,n,s}\! >\! R_{f,n,s}^i]Pr(R_{f,n,s} \!>\! R_{f,n,s}^i)
\end{eqnarray}
\begin{equation}\label{eqn:upper}
\widetilde{V}_{N_R-n+1}(\widetilde{S}_{f,n}) \leq  \widetilde{V}_{N_R-n+1}(\widetilde{S}_{f,n}^*)+ \sum_{\{(i,s)|\forall \mathbf{B}^i_{f,s}(\widetilde{S}_{f,n}) = 0\}}  \bigg( \widetilde{V}_{N_R-n+1}(\widetilde{S}_{f}^{i,s})-\widetilde{V}_{N_R-n+1}(\widetilde{S}_{f,n}^*)\bigg)
\end{equation}
\begin{equation}\label{eqn:lower}
\widetilde{V}_{N_R-n+1}(\widetilde{S}_{f,n}) \geq  \widetilde{V}_{N_R-n+1}(\widetilde{S}_{f,n}^*)+ \sum_{{\{(i,s)|\forall \mathbf{B}^i_{f,s}(\widetilde{S}_{f,n}) = 0\}}} \bigg( \widetilde{V}_{1}(\widetilde{S}_{f}^{i,s})-\widetilde{V}_{1}(\widetilde{S}_{f,n}^*)\bigg)
\end{equation}	
\end{figure*}

\subsection{Online Control}

With the value function $\widetilde{V}_{N_R-n+1}(\widetilde{S}_{f}^*)$ and $\widetilde{V}_{N_R-n+1}(\widetilde{S}_{f}^{i,s})$, the value function for arbitrary system state in arbitrary transmission stage can be approximated via (\ref{eqn:app-v}). Hence the online control action for arbitrary system state $ S_{f,n} $, denoted as $ \Omega_{f,n}^{*}({S}_{f,n}) $, can be obtained the following optimization problem.

\begin{Problem}[Online Optimization]\label{prob:online}
\begin{eqnarray}
\Omega_{f,n}^{*}({S}_{f,n})&=\arg\min& \sum_s  g_{f,n,s}({S}_{f,n},\Omega_{f,n})  + \nonumber\\
&&\sum\limits_{N}\frac{(\lambda_f T_{f,n})^N}{N!} e^{-\lambda_f T_{f,n}}{\widetilde{V}_{N}(\widetilde{S}_{f,n+1})}\nonumber\\
&s.t.& \mbox{(\ref{eqn:cont})}, \ \forall s. \nonumber
\end{eqnarray}

\end{Problem}

Since the value function $ \widetilde{V}_{N}(\widetilde{S}_{f,n+1}) $ is approximated by (\ref{eqn:app-v}), the optimization in Problem \ref{prob:online} can be further decoupled for each segment. For the $ s $-th segment ($ \forall s $), the solution of Problem \ref{prob:online} can be obtained by the following problem, given that the requesting user cannot find the segment from nearby cache nodes.
\begin{Problem}[Online Optimization for the $ s $-th Segment]\label{prob:online-s}
	\begin{eqnarray}
	&&\{P_{f,n,s}^{*}, N_{f,n,s}^{*}\}\nonumber\\
	&=\arg\min& g_{f,n,s}({S}_{f,n},\Omega_{f,n})  + \sum\limits_{N}\frac{(\lambda_f T_{f,n})^N}{N!} e^{-\lambda_f T_{f,n}} \nonumber\\
	&&\sum\limits_{\{i|\forall \mathcal{B}^i_{f,s}(\widetilde{S}_{f,n+1})=0\}} \widetilde{V}_{N}(\widetilde{S}_{f}^{i,s}) - \widetilde{V}_{N}(\widetilde{S}_{f}^{*})\nonumber\\
	&s.t.& \mbox{(\ref{eqn:cont})}, \nonumber
	\end{eqnarray}
where $ \widetilde{S}_{f,n+1} $ is the next system state, and $ \mathcal{B}^i_{f,s}(\widetilde{S}_{f,n+1}) $ represents the buffer status for the $ (f,s) $-th segment in the $ i $-th cache node.
\end{Problem}

Due to the second term of objective in Problem \ref{prob:online-s}, the BS should first choose the cache nodes for downlink receiving, in addition to the requesting user. Based on the selection, the optimal power and transmission time can be derived. Note that this is an integrated continuous and discrete optimization, its solution is summarized below.

\begin{Solution}
	
Given the system state $ S_{f,n} $, let $ d_1,d_2,.. $ be the indexes of cache nodes, whose large-scale attenuation to the BS in the $ s $-th segment is worse than the $ (f,n) $-th user. Moreover, without loss of generality, it is assumed that $ \rho_{d_1}\eta_{f,n,s}^{d_1} \leq  \rho_{d_2}\eta_{f,n,s}^{d_2} \leq ... \leq \rho_{f,n,s}\eta_{f,n,s}$. The optimal control action for the $ s $-th segment ($ \forall s $) can be obtained below.
\begin{itemize}
	\item For each $ i $, suppose the $ d_i $-th cache node are involved for downlink receiving, the optimal power and transmission time control is given by
	\begin{eqnarray}
	&&Q_{d_i,s}^{*}({S}_{f,n})\nonumber\\
	&=\min\limits_{P_{f,n,s} \atop N_{f,n,s}} &\!\!\! g_{f,n,s}({S}_{f,n},\Omega_{f,n})  + \sum\limits_{N}\frac{(\lambda_f T_{f,n})^N}{N!} e^{-\lambda_f T_{f,n}} \nonumber\\
	&&\sum\limits_{j=\{d_1,...,d_{i-1}\}} \widetilde{V}_{N}(\widetilde{S}_{f}^{j,s}) - \widetilde{V}_{N}(\widetilde{S}_{f}^{*})\nonumber\\
	& s.t. &R_{f,n,s}^{d_i} = R_f/N_S. \nonumber
	\end{eqnarray}
	The optimal solution, denoted as $[P_{f,n,s}^{d_i},N_{f,n,s}^{d_i}]$, can be derived similar to Lemma \ref{Lemma:Segment}. 
	\item Let $
d^{*}=\arg \min\limits_{d_i} Q_{d_i,s}^* $, the solution of Problem \ref{prob:online-s} is then given by $
[P_{f,n,s}^{*},N_{f,n,s}^{*}]=[P_{f,n,s}^{d^*},N_{f,n,s}^{d^*}]. $
	
\end{itemize}

\end{Solution}

\subsection{Bound on Value Function Approximation}
In this section, we shall provide the bound on gap between the approximated value function and the actual value function. First of all, we introduce the following bounds on the actual value function.

\begin{Lemma}[Bounds of Value Function] \label{lem:bound} The upper-bound in (\ref{eqn:upper}) holds for value function $ \widetilde{V}_{N_R-n+1}(\widetilde{S}_{f,n}) $ ($ \forall n $). Moreover, if there is no overlap in the service region of cache nodes, the lower-bound in (\ref{eqn:lower}) also holds for  $ \widetilde{V}_{N_R-n+1}(\widetilde{S}_{f,n}) $.
\end{Lemma}
\begin{proof}
Please refer to Appendix B.
\end{proof}

Notice that the proposed linear approximation on value function is actually the upper-bound in (\ref{eqn:upper}), the gap between the approximated value function and the actual value function $ \widetilde{V}_{N_R-n+1}(\widetilde{S}_{f,n}) $, denoted as $ \mathcal{E}_{N_R-n+1}(\widetilde{S}_{f,n}) $, is given by
\begin{eqnarray} 
&&\mathcal{E}_{N_R-n+1}(\widetilde{S}_{f,n}) \leq  \sum_{\{(i,s)|\forall \mathcal{B}^i_{f,s}(\widetilde{S}_{f,n}) = 0\}} \nonumber\\
&&\bigg\{\widetilde{V}_{N_R-n+1}(\widetilde{S}_{f}^{i,s})-\widetilde{V}_{N_R-n+1}(\widetilde{S}_{f}^*)	-\widetilde{V}_{1}(\widetilde{S}_{f}^{i,s})+\widetilde{V}_{1}(\widetilde{S}_{f}^*)\bigg\}.\nonumber
\end{eqnarray}
According to the definition of value function, the average system cost on the $ f $-th file with optimal control can be written as
\begin{equation}
\overline{g}_f^{*} = \sum_{N_R}\frac{(\lambda_f T)^{N_R}}{N_R!} e^{-\lambda_f T} \widetilde{V}_{N_R} (\widetilde{S}_f^0), \label{eqn:system}
\end{equation}
where $ \widetilde{S}_f^0 $ denotes the system state with empty buffer in all cache nodes. When the linear approximation in (\ref{eqn:app-v}) is used, we can evaluate $ \overline{g}_f^{*} $ analytically with at most $\sum_{N_R} \frac{(\lambda_f T)^{N_R}}{N_R!} e^{-\lambda_f T} \mathcal{E}_{N_R} (\widetilde{S}_f^0) $ error.

\section{Simulation}\label{sec:sim}

In the simulation, the radius of one cell is $ 500 $ meters, cache nodes are randomly deployed on the cell-edge region with a service radius of $ 90 $ meters. The number of antennas at the BS is 8. The downlink path loss exponent is $3.5$. The file size is $ 140 $Mb, and is further divided into $ 10 $ segments. The transmission bandwidth is $ 20 $MHz. The weights on transmission energy and time are $w_e=1$ and $w_t=100$ respectively. In addition to the proposed algorithm, the performance of the following two baseline schemes is also compared.

\begin{Baseline}
The BS only ensures the segment delivery to the requesting users in each transmission. The cache nodes with better channel condition to the BS can also decode the file segments.
\end{Baseline}
\begin{Baseline}The BS ensures that all the cache nodes can decode the downlink file in the first transmission. Hence, all the cache nodes can help to forward the file since the second file request.
\end{Baseline}

\begin{figure}[tb]
	\centering
	\includegraphics[scale = 0.6]{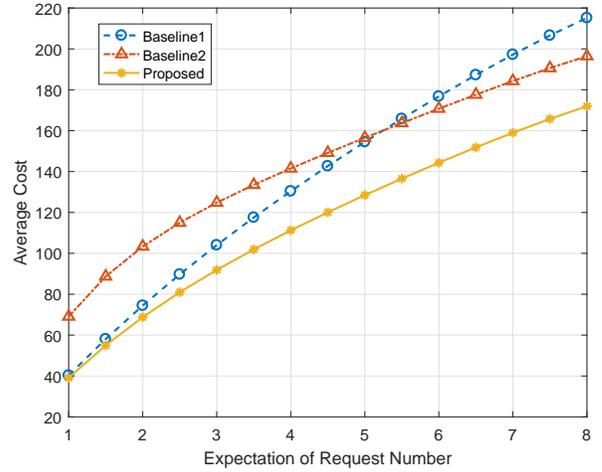}
	\caption{The average total cost versus the expectation of request times, where the number of cache nodes is 20.}
	\label{fig:20caches}
\end{figure}
\begin{figure}[tb]
	\centering
	\includegraphics[scale = 0.6]{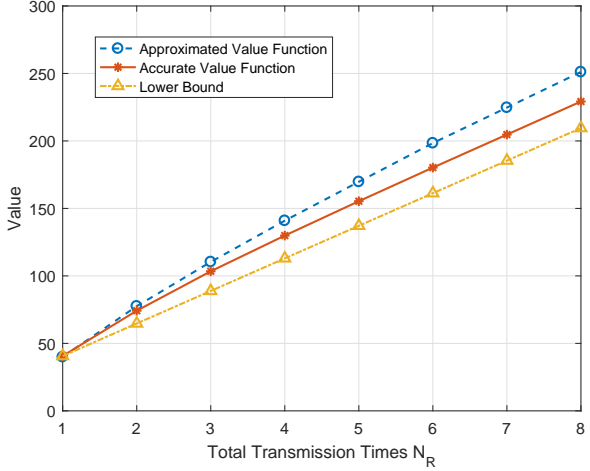}
	\caption{Illustration of value function and its bounds.}
	\label{fig:bound_full}
\end{figure}

The performance of the proposed low-complexity algorithm is compared with the above two baselines in Fig.\ref{fig:20caches}, where the number of cache nodes is $ 20 $. It can be observed that the proposed algorithm is superior to the two baselines for any expected number of requests per file lifetime. Moreover, the Baseline 1 has better performance than Baseline 2 when the popularity of the file is high (larger expected number of file requests).

The approximation error on the value function is illustrated in Fig.\ref{fig:bound_full}, where the value function (average remaining cost) for initial system state $\widetilde{S}^0_{f}$ is plotted with different $ N_R $ (total number of stages). Both upper and lower bounds derived in Lemma \ref{lem:bound} are plotted with the actual value function. It is shown that both bounds are tight, and the approximation error is small. 

\section{Conclusion}\label{sec:con}
We consider the downlink file transmission with the assistance of cache nodes in this paper. Specifically, the number of requests of one file within its lifetime is modeled as a Poisson point process, and the downlink resource minimization problem can be formulated as a Markov decision process with random number of stages. We first propose a revised Bellman's equation, where the optimal control policy can be derived. In order to address the curse of dimensionality, we also introduce a low-complexity sub-optimal solution based on linear approximation of value function. The approximated value function can be calculated analytically, so that conventional value iteration can be eliminated. Finally, we derive a bound on the gap between the approximated value function and the real value function. It is shown by numerical simulation that the proposed algorithm (with the proposed approximated MDP approach) can significantly reduce the resource consumption at the BS.

\section*{Appendix A: Proof Of Lemma \ref{Lemma:Segment} }
First of all, we have the following approximation on the throughput $ R_{f,n,s} $.
\begin{eqnarray}
R_{f,n,s} &\approx& N_{f,n,s} \mathbb{E}_{\mathbf{h}_{f,n,s}} \left[ \log_2 \left(  \frac{||\mathbf{h}_{f,n,s}||^2 P_{f,n,s}}{N_T \sigma^2_z} \right) \right]\nonumber \\
&&=N_{f,n,s}[\theta+\log_2(P_{f,n,s})]\nonumber. 
\end{eqnarray}
With $R_{f,n,s}=R_f/N_S$, we  have $
N_{f,n,s}=\frac{R_f}{N_S[\theta+\log_2(P_{f,n,s})]}$.
Hence the original optimization becomes
\begin{eqnarray}
&\min \limits_{P_{f,n,s} }  &f(P_{f,n,s})= \min \limits_{P_{f,n,s} }  \frac{R_f(w_eP_{f,n,s}+w_t)}{N_S[\theta+\log_2(P_{f,n,s})]} \nonumber
\end{eqnarray}   
Taking first-order derivative on $f(P_{f,n,s})$, the optimal transmission power $ P_{f,n,s}^* $ can be obtained.

\section*{Appendix B: Proof Of Lemma \ref{lem:bound} }

Compared with the cost $ \widetilde{V}_{N_R-n+1}(\widetilde{S}_{f}^*) $ , the additional cost in $ \widetilde{V}_{N_R-n+1}(\widetilde{S}_{f,n}) $ comes from the following two cases: (1) users fall into the coverage of certain cache node, but this cache node does not have the desired file segments; (2) users fall into the region without cache nodes, but the BS would like to spend more transmission resource so that some cache nodes can decode the downlink segments together with the users. 

Note that when the BS is transmitting file segments to one cache nodes, some other cache nodes with better channel condition can also decode the segments, we have $ \sum\limits_{\{(i,s)|\forall \mathbf{B}^i_{f,s}(\widetilde{S}_{f,n}) = 0\}}  \bigg( \widetilde{V}_{N_R-n+1}(\widetilde{S}_{f}^{i,s})-\widetilde{V}_{N_R-n+1}(\widetilde{S}_{f,n}^*)\bigg) $ is greater than $ \widetilde{V}_{N_R-n+1}(\widetilde{S}_{f,n})-\widetilde{V}_{N_R-n+1}(\widetilde{S}_{f}^*) $.

Moreover, the additional cost $  \widetilde{V}_{1}(\widetilde{S}_{f}^{i,s})-\widetilde{V}_{1}(\widetilde{S}_{f,n}^*) $ comes from the extra transmission resource of the BS, which is exactly to serve one user in the coverage of the $ i $-th cache node. Thus the lower-bound is also straightforward.

\bibliographystyle{IEEEtran}
\bibliography{icc}

\end{document}